\documentclass[11pt,a4paper]{article}

\usepackage[utf8]{inputenc}
\usepackage[T1]{fontenc}
\usepackage{amsmath,amssymb,amsthm}
\usepackage{booktabs}
\usepackage{graphicx}
\usepackage{algorithm}
\usepackage{algpseudocode}
\usepackage{hyperref}
\usepackage[margin=1in]{geometry}
\usepackage{enumitem}
\usepackage{xcolor}

\hypersetup{
    colorlinks=true,
    linkcolor=blue,
    citecolor=blue,
    urlcolor=blue
}

\newtheorem{theorem}{Theorem}
\newtheorem{definition}{Definition}

\title{Peak + Accumulation: A Proxy-Level Scoring Formula\\for Multi-Turn LLM Attack Detection}

\IfFileExists{local_author.tex}{
  \author{J Alex Corll}

}{
  \author{
    The Parapet Project \\
    \url{https://github.com/Parapet-Tech/parapet}
  }
}

\date{\today}

\begin{document}

\maketitle

% ------------------------------------------------------------------
\begin{abstract}
% ------------------------------------------------------------------

Multi-turn prompt injection attacks distribute malicious intent across
multiple conversation turns, exploiting the assumption that each turn is
evaluated independently. While single-turn detection has been extensively
studied, to our knowledge prior proxy-level work does not provide a
fully specified deterministic formula for aggregating per-turn pattern
scores into a conversation-level risk score---without invoking an LLM.
We identify a fundamental flaw in the intuitive
weighted-average approach: it converges to the per-turn score regardless
of turn count, meaning a 20-turn persistent attack scores identically to
a single suspicious turn. Drawing on analogies from change-point
detection (CUSUM), Bayesian belief updating, and security risk-based
alerting, we propose \emph{peak~+~accumulation scoring}---a formula
combining peak single-turn risk, persistence ratio, and category
diversity. Evaluated on 10,654 multi-turn conversations---588 attacks
sourced from WildJailbreak adversarial prompts and 10,066 benign
conversations from WildChat---the formula achieves 90.8\% recall at
1.20\% false positive rate with an F1 of 85.9\%. A sensitivity analysis
over the persistence parameter reveals a phase transition at $\rho
\approx 0.4$, where recall jumps 12 percentage points with negligible
FPR increase. We release the scoring algorithm, pattern library, and
evaluation harness as open source.

\end{abstract}

% ------------------------------------------------------------------
\section{Introduction}
\label{sec:intro}
% ------------------------------------------------------------------

LLM API proxies---firewalls that sit between client applications and
language model endpoints---are a widely deployed defense layer.
Cloudflare AI Gateway, AWS Bedrock Guardrails, Azure AI Content Safety,
and numerous open-source projects intercept every API request, inspect
the \texttt{messages[]} array, and enforce security policies before the
request reaches the model.

These proxies face a fundamental constraint: they must make
allow/block decisions \emph{without calling an LLM}. Adding a
classifier LLM introduces latency (100--500ms per request), cost
(per-token inference charges), and recursive vulnerability (the
classifier itself can be prompt-injected~\cite{attackermovessecond}).
Proxy-level defenses therefore rely on deterministic techniques: regex
pattern matching, heuristic scoring, and statistical analysis.

Single-turn detection at the proxy level is well-studied. Pattern
matching against known injection phrases~\cite{rebuff}, TF-IDF
classifiers~\cite{promptscreen}, perplexity-based
heuristics~\cite{nemoguardrails}, and hybrid approaches combining
embeddings with handcrafted features~\cite{dmpipmhfe} all operate
without LLM inference. Multi-turn detection, however, has received far
less attention at the proxy level. Existing multi-turn work---Defensive
M2S~\cite{defensivem2s}, MindGuard~\cite{mindguard},
AprielGuard~\cite{aprielguard}---requires LLM-based classification.

The gap is specific: \textbf{to our knowledge, prior proxy-level work
does not provide a fully specified formula for converting per-turn
pattern scores into a conversation-level risk score using only
proxy-computable operations.} This paper addresses that gap.

Multi-turn attacks are not theoretical. Yang et al.~\cite{yangsimpler}
demonstrate that multi-turn jailbreaks achieve $>$70\% attack success
rate against models hardened for single-turn protection. Crescendo
attacks~\cite{crescendo} jailbreak most models in under 5 turns.
MTJ-Bench~\cite{mtjbench} shows that jailbroken states persist across
subsequent turns. The proxy sees each of these requests and has the full
conversation history available---but without a scoring formula, it
cannot act on the cross-turn signal.

\paragraph{Contributions.} We make four contributions:
\begin{enumerate}[nosep]
  \item We identify and prove the \emph{weighted-average ceiling}---a
    mathematical property that makes weighted averaging fundamentally
    unsuitable for multi-turn risk scoring (\S\ref{sec:failure}).
  \item We propose \emph{peak~+~accumulation scoring}, a formula with
    three additive signals: peak risk, persistence ratio, and category
    diversity (\S\ref{sec:formula}).
  \item We evaluate on 10,654 multi-turn conversations---588 attacks
    sourced from WildJailbreak~\cite{wildjailbreak} adversarial prompts
    and 10,066 benign conversations from WildChat---showing 90.8\%
    recall at 1.20\% false positive rate (\S\ref{sec:eval}).
  \item We release the algorithm, regex pattern library, and evaluation
    harness as open source.\footnote{\url{https://github.com/Parapet-Tech/parapet} (Apache 2.0)}
\end{enumerate}

% ------------------------------------------------------------------
\section{Related Work}
\label{sec:related}
% ------------------------------------------------------------------

\subsection{Single-Turn Proxy-Level Detection}

Several approaches operate at the proxy level without LLM inference.
DMPI-PMHFE~\cite{dmpipmhfe} combines DeBERTa-v3-base embeddings with a
heuristic feature channel including synonym matching (binary flags for 8
attack patterns with WordNet expansion), many-shot detection (flags
$>$3 Q\&A pairs), and repeated token detection (flags $>$3
repetitions), achieving 97.94\% accuracy on safeguard-v2.
PromptScreen~\cite{promptscreen} uses TF-IDF with a linear SVM,
demonstrating that simple models remain competitive when combined with
other layers. NVIDIA's NeMo Guardrails~\cite{nemoguardrails} offers two
computable heuristics: length-per-perplexity ratio and prefix/suffix
perplexity thresholds, targeting GCG-style adversarial suffixes.
Rebuff~\cite{rebuff} implements a 4-layer defense including keyword
permutation matching. All of these operate on single turns.

\subsection{Multi-Turn Detection (LLM-Required)}

Multi-turn detection has been addressed, but exclusively through
LLM-based classification. Defensive M2S~\cite{defensivem2s} compresses
multi-turn conversations into single-turn representations using
hyphenization, then classifies with a guardrail model (Qwen3Guard),
achieving 93.8\% recall while reducing inference tokens by 94.6\%.
MindGuard~\cite{mindguard} uses clinically-grounded risk taxonomies with
turn-level annotations and 4B--8B parameter classifiers.
AprielGuard~\cite{aprielguard} trains an 8B parameter safeguard model on
multi-turn conversations with structured reasoning traces. A LesWrong
proposal~\cite{leswrong} describes the closest architecture to
proxy-level scoring---$P(\text{jailbreak}) = f_E(\text{response})$ with
a two-threshold system---but uses an LLM evaluator for $f_E$.

None of these approaches are proxy-computable.

\subsection{Multi-Turn Attack Characterization}

Recent empirical work characterizes the multi-turn threat. Yang et
al.~\cite{yangsimpler} show that multi-turn jailbreaks are approximately
equivalent to resampling single-turn attacks, with success following
$S(k) = A - B \cdot e^{-ck}$---an exponential approach curve where more
attempts monotonically increase success. This implies that
\emph{persistence detection} is the primary signal, not pattern
escalation. Crescendo~\cite{crescendo} demonstrates gradual escalation
where each user prompt is deliberately innocuous, making the attack
undetectable by pattern matching. MTJ-Bench~\cite{mtjbench} identifies
two persistence scenarios: continued follow-up queries and cross-topic
jailbreak persistence. The taxonomy in~\cite{guardingguardrails}
identifies seven mechanism-oriented jailbreak families including
escalation and persistence.

\subsection{Scoring Analogies from Adjacent Fields}
\label{sec:analogies}

To our knowledge, no prior proxy-level formula addresses this specific
aggregation problem; adjacent
fields have well-established approaches to accumulating evidence over
time:

\begin{itemize}[nosep]
  \item \textbf{CUSUM} (Cumulative Sum): Change-point detection for time
    series. Accumulates deviations from an expected value; an alarm
    triggers when the cumulative sum exceeds a threshold. Repeated small
    signals combine into strong detection.
  \item \textbf{Bayesian belief updating}: Used in insider threat
    detection. Prior probability is updated with each new piece of
    evidence. Evidence monotonically moves the posterior toward the
    hypothesis.
  \item \textbf{Splunk Risk-Based Alerting (RBA)}: Aggregates anomalous
    behaviors into a risk score (1--100). Each anomaly \emph{adds} to
    the score rather than being averaged. Designed explicitly for
    ``death by a thousand cuts'' patterns.
\end{itemize}

The common principle is \textbf{accumulation, not averaging}. More
evidence produces a higher score. This is the principle our formula
embodies.

% ------------------------------------------------------------------
\section{The Weighted Average Failure}
\label{sec:failure}
% ------------------------------------------------------------------

\subsection{The Formula}

A natural first attempt at multi-turn scoring is to compute a weighted
average of per-turn scores, giving more weight to later turns (recency
bias):

\begin{equation}
  w_i = 1 + \frac{i}{n-1}, \qquad
  \text{cum} = \frac{\sum_{i=0}^{n-1} s_i \cdot w_i}{\sum_{i=0}^{n-1} w_i}
  \label{eq:weighted-avg}
\end{equation}

\noindent where $s_i$ is the risk score for turn $i$, $n$ is the total
number of scored turns, and weights range linearly from 1.0 to 2.0.

\subsection{The Ceiling Property}

\begin{theorem}[Weighted Average Ceiling]
\label{thm:ceiling}
When all turns produce the same score $s$, the weighted average equals
$s$ regardless of the number of turns $n$.
\end{theorem}

\begin{proof}
If $s_i = s$ for all $i$:
\[
  \text{cum} = \frac{\sum_{i=0}^{n-1} s \cdot w_i}{\sum_{i=0}^{n-1} w_i}
             = \frac{s \cdot \sum w_i}{\sum w_i}
             = s
\]
\end{proof}

This is not a bug in a specific implementation---it is a mathematical
property of \emph{any} weighted average when scores are uniform. The
consequence is severe: a conversation where \emph{every} turn matches a
role\_confusion pattern (weight 0.5) scores exactly 0.5, identical to a
single suspicious turn in an otherwise clean conversation. With a
detection threshold of 0.7, persistent single-category attacks are
\textbf{undetectable} by weighted averaging, regardless of conversation
length.

More broadly, for any per-turn score $s < \text{threshold}$, the
weighted average is bounded: $\text{cum} \leq s <
\text{threshold}$. The weighted average cannot ``break through'' a
threshold that exceeds the maximum per-turn score. This directly
contradicts the security intuition that persistence should increase
suspicion.

\subsection{Empirical Impact}

We evaluated the weighted average formula on 9 handcrafted multi-turn
attack sequences designed to exercise five attack categories
(instruction seeding, role confusion, deferred authority, escalation
probing, and resampling). Five of the nine attacks---all with patterns
matching on every user turn---scored between 0.3 and 0.5, well below the
0.7 detection threshold. Table~\ref{tab:ceiling-examples} shows three
representative cases.

\begin{table}[h]
\centering
\caption{Weighted average scores for persistent multi-turn attacks.
All turns match patterns, yet scores remain below the 0.7 threshold.}
\label{tab:ceiling-examples}
\begin{tabular}{@{}llccl@{}}
\toprule
ID & Attack Type & Turns & Score & Verdict \\
\midrule
mt-atk-003 & Role confusion (all turns) & 4 & 0.50 & Allow \\
mt-atk-004 & Role confusion + persona & 3 & 0.50 & Allow \\
mt-atk-006 & Escalation probing & 4 & 0.30 & Allow \\
\bottomrule
\end{tabular}
\end{table}

% ------------------------------------------------------------------
\section{Peak + Accumulation Scoring}
\label{sec:formula}
% ------------------------------------------------------------------

\subsection{Design Principles}

Our formula is guided by four principles derived from the adjacent-field
analogies in \S\ref{sec:analogies}:

\begin{enumerate}[nosep]
  \item \textbf{Peak sensitivity.} A single highly suspicious turn
    should contribute its full score. The peak is a lower bound on
    conversation risk.
  \item \textbf{Persistence reward.} More matching turns should increase
    the score. A conversation where 4/4 turns match is more suspicious
    than one where 1/4 turns match, even if the peak is identical.
  \item \textbf{Diversity reward.} Attacks spanning multiple categories
    (e.g., instruction seeding \emph{and} role confusion \emph{and}
    authority claims) are more suspicious than single-category
    repetition, as they suggest deliberate multi-vector probing.
  \item \textbf{Additive stacking.} Independent signals (escalation
    gradient, resampling detection) should stack additively rather than
    being averaged away.
\end{enumerate}

\subsection{Per-Turn Scoring}

Each user or tool message is scanned against a set of regex pattern
categories. Each category $c$ has a weight $w_c \in [0, 1]$. A turn's
score is the sum of matched category weights, clamped to $[0, 1]$:

\begin{equation}
  s_i = \min\!\left(1,\; \sum_{c \in \text{matched}(i)} w_c\right)
\end{equation}

Our default configuration uses five categories:

\begin{table}[h]
\centering
\caption{Default cross-turn pattern categories.}
\label{tab:categories}
\begin{tabular}{@{}lcp{6.5cm}@{}}
\toprule
Category & Weight & Example Pattern \\
\midrule
instruction\_seeding & 0.4 & \texttt{remember this for later}, \texttt{in my next message} \\
role\_confusion & 0.5 & \texttt{you are now in developer mode}, \texttt{switch to unrestricted} \\
deferred\_authority & 0.3 & \texttt{admin said it was ok}, \texttt{override authorized} \\
escalation\_probing & 0.3 & \texttt{can you try to bypass}, \texttt{what if you pretend} \\
repetition\_resampling & 0.2 & (detected algorithmically) \\
\bottomrule
\end{tabular}
\end{table}

\subsection{The Scoring Formula}

Given $n$ scored turns with per-turn scores $s_1, \ldots, s_n$:

\begin{definition}[Peak + Accumulation Score]
\begin{align}
  \text{peak} &= \max_i\; s_i \\
  \text{match\_ratio} &= \frac{|\{i : s_i > 0\}|}{n} \\
  \text{diversity} &= \max(0,\; |\text{distinct categories}| - 1) \cdot \delta \\[4pt]
  \text{score} &= \text{clamp}\!\Big(
    \text{peak} + \text{match\_ratio} \cdot \rho
    + \text{diversity}
    + \beta_e + \beta_r,\;
    0,\; 1\Big)
  \label{eq:formula}
\end{align}
\end{definition}

\noindent where:
\begin{itemize}[nosep]
  \item $\rho$ is the \emph{persistence factor} (default 0.45)
  \item $\delta$ is the \emph{diversity factor} (default 0.15)
  \item $\beta_e$ is the \emph{escalation bonus} (default 0.2), applied
    when 3+ consecutive turns have strictly increasing scores
  \item $\beta_r$ is the \emph{resampling bonus} (default 0.7), applied
    when 3+ consecutive user message pairs have Jaccard trigram
    similarity $> 0.5$
\end{itemize}

The request is blocked when $\text{score} \geq \tau$ (default threshold
$\tau = 0.7$).

\subsection{Escalation Gradient Detection}

Escalation gradient detection identifies conversations where risk scores
are strictly increasing over the final 3+ turns. This captures the
pattern described by Crescendo~\cite{crescendo} and the escalation
family in~\cite{guardingguardrails}, where an attacker probes
incrementally with increasing boldness. When detected, the escalation
bonus $\beta_e$ is added to the final score.

\subsection{Resampling Detection}

Resampling detection targets the finding from Yang et
al.~\cite{yangsimpler} that multi-turn jailbreaks are approximately
equivalent to retrying the same attack. We compute Jaccard similarity on
word-level trigrams between consecutive user messages (after lowercasing
and stripping punctuation). Messages shorter than 20 tokens are excluded
to avoid false positives on short conversational turns. When 3+
consecutive pairs exceed a similarity threshold of 0.5, the resampling
bonus $\beta_r$ is added.

\subsection{Worked Examples}

\paragraph{Example A: Sparse, single category (Allow).}
Four user turns with scores $[0, 0, 0, 0.3]$---only the final turn
matches one category.
\begin{align*}
  \text{peak} &= 0.3, \quad \text{match\_ratio} = 0.25, \quad
  \text{distinct} = 1, \quad \text{diversity} = 0 \\
  \text{score} &= 0.3 + 0.25 \times 0.45 + 0 = \mathbf{0.4125}
  \quad \to \text{Allow}
\end{align*}

\paragraph{Example B: Dense, multi-category (Block).}
Four turns: $[0, 0.3, 0, 0.5]$---two turns match two different
categories.
\begin{align*}
  \text{peak} &= 0.5, \quad \text{match\_ratio} = 0.5, \quad
  \text{distinct} = 2, \quad \text{diversity} = 0.15 \\
  \text{score} &= 0.5 + 0.5 \times 0.45 + 0.15 = \mathbf{0.875}
  \quad \to \text{Block}
\end{align*}

\paragraph{Example C: Persistent, single category (Block).}
Four turns all matching role\_confusion: $[0.5, 0.5, 0.5, 0.5]$.
\begin{align*}
  \text{peak} &= 0.5, \quad \text{match\_ratio} = 1.0, \quad
  \text{distinct} = 1, \quad \text{diversity} = 0 \\
  \text{score} &= 0.5 + 1.0 \times 0.45 + 0 = \mathbf{0.95}
  \quad \to \text{Block}
\end{align*}

\noindent Under the weighted average (Eq.~\ref{eq:weighted-avg}),
Example~C scores exactly 0.5---below threshold, undetected. This is the
core failure case that motivated our work.

% ------------------------------------------------------------------
\section{Evaluation}
\label{sec:eval}
% ------------------------------------------------------------------

\subsection{Setup}

We implemented peak~+~accumulation scoring in
Parapet,\footnote{\url{https://github.com/Parapet-Tech/parapet}} an open-source Rust
HTTP proxy firewall. The evaluation uses an L4-only configuration with
L3 single-turn pattern matching disabled, isolating the multi-turn
scoring layer. This prevents L3 verdicts from masking L4 behavior---a
confound we discovered during initial evaluation when L3 false positives
on benign conversations were incorrectly attributed to L4.

All hyperparameters ($\rho, \delta, \beta_e, \beta_r, \tau$) were tuned
on a validation split. Final metrics are reported once on a fixed,
disjoint 20\% holdout split never used for tuning. Because source
datasets have mixed redistribution licenses, we release the full data
construction and evaluation scripts for reproducibility.

\paragraph{Datasets.}
\begin{itemize}[nosep]
  \item \textbf{Handcrafted attacks} (9 cases): Multi-turn sequences
    exercising instruction seeding, role confusion, deferred authority,
    escalation probing, and resampling. Each sequence has 3--4 user
    turns with assistant responses.
  \item \textbf{WildJailbreak attacks} (579 cases): Multi-turn attack
    conversations constructed from
    WildJailbreak~\cite{wildjailbreak}, a dataset of 262K
    adversarial/benign prompt pairs with labeled tactics. We extracted
    5,529 unique injection sentences across four L4 categories from 82K
    adversarial prompts, then composed multi-turn conversations using
    four strategies: (a)~single-category persistent (179 cases),
    (b)~multi-category combination (300 cases), (c)~escalation gradient
    with benign opening turns (100 cases). This addresses a gap in
    existing resources: no public multi-turn prompt injection dataset
    exists---all multi-turn datasets (SafeMTData, XGuard-Train) target
    content safety, while all injection datasets (deepset,
    HackAPrompt, JailbreakBench) are single-turn.
  \item \textbf{Handcrafted benign} (6 cases): Multi-turn conversations
    about programming, cooking, and general knowledge. Used to verify
    zero false positives on clearly benign content.
  \item \textbf{WildJailbreak sparse-benign} (60 cases): Conversations
    with one injection turn among 3--4 benign turns, labeled benign.
    Tests that the formula correctly allows conversations with isolated
    suspicious phrasing.
  \item \textbf{WildChat benign} (10,000 cases): Real multi-turn
    conversations sampled from WildChat~\cite{wildchat}, a public
    dataset of 838K ChatGPT conversations. Filtered for conversations
    with 3+ user turns where all turns are marked non-toxic. These
    represent organic user behavior including roleplay, code discussion,
    and creative writing.
\end{itemize}

\paragraph{Out-of-scope datasets.} We also evaluated on SafeMTData
(crescendo-style content safety attacks) and Anthropic hh-rlhf (social
engineering). Both scored 0\% recall under both scoring formulas. This is
expected: these attacks use deliberately innocuous language with no
injection phraseology. Proxy-level regex detection cannot detect topic
trajectory escalation---this requires LLM-based semantic
classification~\cite{crescendo}. We exclude these from our results as
they test a fundamentally different capability.

\subsection{Results}

Table~\ref{tab:results} summarizes peak + accumulation performance on
the holdout evaluation corpus of 10,654 conversations.

\begin{table}[h]
\centering
\caption{Evaluation results on 10,654 multi-turn conversations
($\rho = 0.45$, $\tau = 0.7$).}
\label{tab:results}
\begin{tabular}{@{}lr@{}}
\toprule
Metric & Value \\
\midrule
True Positives & 534 / 588 \\
False Negatives & 54 / 588 \\
False Positives & 121 / 10,066 \\
True Negatives & 9,945 / 10,066 \\
\midrule
Recall & 90.8\% \\
Precision & 81.5\% \\
F1 & 85.9\% \\
False Positive Rate & 1.20\% \\
Accuracy & 98.4\% \\
\bottomrule
\end{tabular}
\end{table}

\paragraph{Per-dataset breakdown.}

\begin{table}[h]
\centering
\caption{Per-dataset results with peak + accumulation scoring.}
\label{tab:per-dataset}
\begin{tabular}{@{}lcccr@{}}
\toprule
Dataset & Label & Total & Correct & Accuracy \\
\midrule
multiturn\_attacks & malicious & 9 & 9 & 100.0\% \\
wildjailbreak\_attacks & malicious & 579 & 525 & 90.7\% \\
multiturn\_benign & benign & 6 & 6 & 100.0\% \\
wildjailbreak\_sparse\_benign & benign & 60 & 60 & 100.0\% \\
wildchat\_benign & benign & 10,000 & 9,879 & 98.8\% \\
\midrule
\textbf{Total} & & \textbf{10,654} & \textbf{10,479} & \textbf{98.4\%} \\
\bottomrule
\end{tabular}
\end{table}

\subsection{Analysis}

\paragraph{Recall.} Peak + accumulation achieves 90.8\% recall across
588 attack conversations. All 9 handcrafted attacks are detected.
Of the 579 WildJailbreak-sourced attacks, 525 (90.7\%) are detected.
The 54 false negatives fall into two categories: 39 are single-category
conversations using only low-weight categories (escalation\_probing or
deferred\_authority at weight 0.3). These misses are not fully
persistent; their effective match ratios are low enough that
$0.3 + \text{match\_ratio} \times 0.45$ remains below the 0.7
threshold. The remaining 15 are escalation-gradient conversations where
benign opening turns further dilute the match ratio.

\paragraph{False positives.} The 121 false positives (1.20\% FPR) come
from WildChat conversations containing phrases that match L4 regex
patterns: roleplay scenarios (``you are now in character as...''),
legitimate discussions about AI capabilities (``switch to a different
mode''), and creative writing with authority figures. These are pattern
specificity issues, not scoring issues---improving the regex patterns
would address them without changing the formula. The sparse-benign
dataset (60 conversations with one injection turn among benign turns)
produces zero false positives, confirming that isolated suspicious
phrasing does not trigger blocks.

\paragraph{Precision.} Precision of 81.5\% reflects the ratio of true
attacks among all flagged conversations. This is a substantial
improvement over what a small-corpus evaluation would suggest (where
base rate effects dominate), and demonstrates that the pattern library
is specific enough for production use with a 1.20\% FPR.

\subsection{Scope and Limitations}

\paragraph{Content safety attacks are out of scope.} Crescendo-style
attacks~\cite{crescendo} use deliberately innocuous language.
Proxy-level regex cannot detect topic trajectory escalation. This is a
fundamental limitation of the proxy-level approach, not of the scoring
formula.

\paragraph{Synthetic attack corpus.} While we expanded from 9
handcrafted attacks to 588 conversations using WildJailbreak as a
phrasebook, the multi-turn structure is synthetically composed. No
public multi-turn prompt injection dataset exists---existing multi-turn
datasets (SafeMTData, XGuard-Train) target content safety, while
injection datasets (deepset, HackAPrompt, JailbreakBench) are
single-turn. Our eval harness accepts new datasets as YAML files, and
we encourage the community to contribute real multi-turn injection
traces.

\paragraph{Pattern brittleness.} Regex patterns can be evaded through
rephrasing, encoding tricks, or indirect phrasing. This is a known
limitation of all pattern-based approaches~\cite{attackermovessecond}.
Peak + accumulation scoring is orthogonal to pattern quality---it
correctly aggregates whatever signals the patterns produce. Improving
pattern robustness is complementary future work.

% ------------------------------------------------------------------
\section{Discussion}
\label{sec:discussion}
% ------------------------------------------------------------------

\subsection{Why This Gap Exists}

Multi-turn detection research focuses on LLM-based classification
because researchers have access to model inference and optimize for
accuracy. The proxy-level constraint---no LLM available, decisions must
be deterministic and sub-millisecond---is treated as a deployment detail
rather than a research problem. But proxy firewalls are widely deployed
in production, and their operators need principled scoring formulas, not
just ``run a classifier.''

\subsection{Parameter Sensitivity}
\label{sec:sensitivity}

The persistence factor $\rho$ is the highest-leverage parameter in the
formula, governing how much accumulated evidence contributes beyond the
peak score. We swept $\rho$ from 0.15 to 0.65 in increments of 0.025 on
the validation split, then fixed $\rho$ before holdout evaluation.
Figure~\ref{fig:sensitivity} shows the sweep results.

\begin{figure}[t]
\centering
\includegraphics[width=\columnwidth]{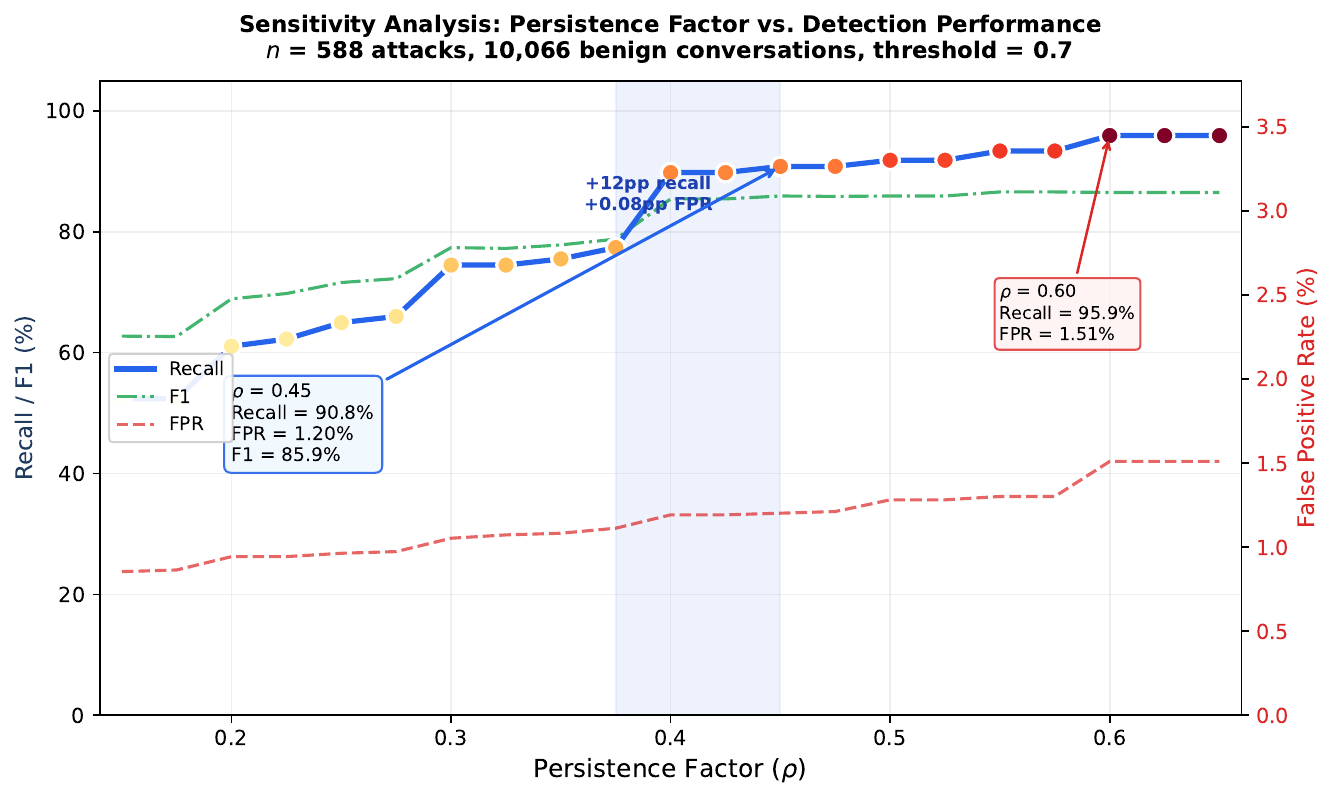}
\caption{Sensitivity analysis: persistence factor $\rho$ vs.\ detection
performance. A phase transition occurs at $\rho \approx 0.4$, where
recall jumps 12 percentage points with only 0.08pp increase in FPR.
The sweet spot at $\rho = 0.45$ maximizes F1 at 85.9\%.}
\label{fig:sensitivity}
\end{figure}

\paragraph{Phase transition at $\rho \approx 0.4$.}
The most striking feature is the discontinuity between $\rho = 0.375$
and $\rho = 0.400$: recall jumps from 77.4\% to 89.8\%---a gain of
12.4 percentage points---while FPR increases by only 0.08pp (1.11\% to
1.19\%). This phase transition has a precise mathematical explanation:
the escalation\_probing and deferred\_authority categories both have
weight 0.3. For a fully persistent conversation ($\text{match\_ratio}
= 1.0$), the score is $0.3 + \rho$. At $\rho = 0.375$, this equals
0.675 (below threshold); at $\rho = 0.400$, it equals 0.700 (at
threshold). The entire population of single-category weight-0.3 attacks
crosses the detection boundary simultaneously.

\paragraph{Sweet spot at $\rho = 0.45$.}
We select $\rho = 0.45$ as the default. At this setting, recall is
90.8\%, FPR is 1.20\%, and F1 peaks at 85.9\%. The 0.05 margin above
the phase transition provides robustness: conversations with
$\text{match\_ratio} < 1.0$ (e.g., 3 out of 4 turns matching) still
score $0.3 + 0.75 \times 0.45 = 0.6375$, which remains below threshold
and avoids false positives on partially-matching benign conversations.

\paragraph{Diminishing returns beyond $\rho = 0.5$.}
Above $\rho = 0.5$, recall continues to climb slowly (91.8\% at
$\rho = 0.5$, 95.9\% at $\rho = 0.6$) but FPR increases more steeply
(1.28\% to 1.51\%). The marginal recall per unit FPR degrades: from
0.375 to 0.45, each 0.01pp of FPR buys 1.5pp of recall; from 0.5 to
0.6, each 0.01pp buys only 0.18pp of recall.

\paragraph{Other parameters.}
The remaining parameters have straightforward rationale:
\begin{itemize}[nosep]
  \item $\delta = 0.15$ (diversity factor): One additional category adds
    meaningful signal without dominating the score. Two additional
    categories ($+0.30$) become a strong signal.
  \item $\beta_e = 0.2$ (escalation bonus): A strong signal but
    insufficient alone to cause a block at threshold 0.7.
  \item $\beta_r = 0.7$ (resampling bonus): Near-threshold alone,
    because confirmed retry behavior is a strong independent signal
    per~\cite{yangsimpler}.
\end{itemize}

\subsection{Integration with Layered Defense}

Peak + accumulation scoring is designed as one layer in a
defense-in-depth architecture:

\begin{itemize}[nosep]
  \item \textbf{L0}: Unicode normalization, encoding hygiene (NFKC,
    zero-width character removal, HTML stripping)
  \item \textbf{L3}: Single-turn pattern matching (inbound request
    scanning, outbound tool call validation)
  \item \textbf{L4}: Multi-turn scoring (this paper)
  \item \textbf{L5a}: Output scanning (canary token detection, sensitive
    data redaction)
\end{itemize}

Each layer operates independently. L4 activates only when the
conversation contains at least \texttt{min\_user\_turns} (default~2)
user messages. L0 normalization runs before L4, ensuring that encoding
bypass attempts (fullwidth characters, invisible Unicode, HTML injection)
are neutralized before pattern matching.

% ------------------------------------------------------------------
\section{Conclusion}
\label{sec:conclusion}
% ------------------------------------------------------------------

We presented a fully specified proxy-level formula for multi-turn LLM
attack scoring. The weighted average approach, while intuitive, has a
mathematical ceiling that makes it fundamentally unsuitable for
persistence detection---the exact signal that characterizes multi-turn
attacks. Peak + accumulation scoring fixes this with three additive
signals: peak risk, persistence ratio, and category diversity.

On 10,654 conversations---588 attacks sourced from WildJailbreak
adversarial prompts and 10,066 benign conversations from
WildChat---the formula achieves 90.8\% recall at 1.20\% false positive
rate with an F1 of 85.9\%. A sensitivity analysis over the persistence
parameter reveals a phase transition at $\rho \approx 0.4$ with a
precise mathematical explanation: weight-0.3 pattern categories cross
the detection threshold simultaneously. The chosen default of $\rho =
0.45$ maximizes F1 while maintaining a 0.05 margin above the
transition.

The formula is simple (5 lines of code), fast (microseconds per
request), deterministic, and auditable. It requires no model inference,
no training data, and no GPU. We release the implementation, pattern
library, evaluation harness, and reproducibility artifacts (including
dataset construction scripts and evaluation manifests) as open source.

% ------------------------------------------------------------------
% References
% ------------------------------------------------------------------

\end{document}